\documentclass[a4paper,onecolumn,11pt,accepted=2023-11-19]{quantumarticle}
\pdfoutput=1
\usepackage{amsfonts,amsmath,amsthm,amssymb,dsfont, etex,bm,mathtools}
\usepackage[colorlinks=true,linkcolor=blue,citecolor=blue]{hyperref}
\usepackage{paralist}
\usepackage{enumerate}
\usepackage{framed}
\usepackage{mdframed}

\usepackage{tikz}
\usetikzlibrary{arrows,automata}
\usepackage{comment}
\usepackage{graphicx,float,wrapfig}
\usepackage{mathrsfs}
\usepackage[usenames,dvipsnames]{pstricks}
\usepackage{epsfig}
\usepackage{pst-grad} % Fdi gradients
\usepackage{pst-plot} % Fdi axes
\usepackage{mathdots}
\usepackage[linesnumbered,boxed,ruled,vlined]{algorithm2e}

\usepackage{bbm}
\usepackage{url}
\usepackage{tabu}
\usepackage{xfrac}
\usepackage[normalem]{ulem}
\usepackage{bm}
\usepackage{mleftright}

\usepackage{hyperref}  
\hypersetup{colorlinks=true,citecolor=blue,linkcolor=blue}

\usepackage[capitalize]{cleveref}

\usepackage[breakable]{tcolorbox}
\newtcolorbox{construction}[2][]
{
	breakable,
	colframe = gray!50,
	colback  = gray!10,
	coltitle = gray!10!black,
	before skip = 10pt,
	after skip = 10pt,
	title    = \textbf{#2},
	#1,
}
\newtcolorbox{graphview}[2][]
{
	breakable,
	colframe = black!30,
	colback  = black!0,
	coltitle = gray!10!black,
	before skip = 10pt,
	after skip = 10pt,
	title    = \textbf{#2},
	#1,
}

\renewcommand{\epsilon}{\varepsilon}
\newcommand{\Ex}{\mathop{\mathbb{E}}}

\newcommand{\poly}{\mathop{\mathrm{poly}}}
\newcommand{\polylog}{\mathop{\mathrm{polylog}}}

\newcommand{\N}{\mathbb{N}}

\newcommand{\bits}{\{0,1\}}

\newcommand{\QMA}{\mathsf{QMA}\xspace}

\newcommand{\NP}{\mathsf{NP}\xspace}

\def \ro {\text{r.o.-}}

\newcommand{\cro}[1]{$c$\text{-}\ro}

\newcommand{\PCP}{\mathsf{PCP}\xspace}

\newcommand{\GF}{\mathsf{GF}}

\newcommand{\negl}{\mathsf{negl}}

\newcommand{\ie}{\textit{i}.\textit{e}.\@\xspace}
\newcommand{\eg}{\textit{e}.\textit{g}.\@\xspace}

\def\poly{\mathrm{poly}}
\def\polylog{\mathrm{polylog}}

\def\tt{\mathsf{tt}}

\def\caH{\mathcal{H}}
\def\caP{\mathcal{P}}

\def\caG{\mathcal{G}}

\usepackage{fullpage}
\newcommand{\makeName}[1]{%
	\expandafter\newcommand\csname#1\endcsname{\mathsf{#1}}}

\usepackage{makecell}
\usepackage{aliascnt} %For correct autoref with shared counters. Not necessary if counters are not shared.

\newtheorem{theorem}{Theorem}[section]

\newtheorem{open}{Open Problem}

\newtheorem*{theorem*}{Theorem}

\newaliascnt{definition}{theorem}
\theoremstyle{definition}
\newtheorem{definition}[definition]{Definition}
\aliascntresetthe{definition}

\newtheorem*{definition*}{Definition}

\theoremstyle{plain}
\newaliascnt{lemma}{theorem}

\aliascntresetthe{lemma}

\newtheorem*{lemma*}{Lemma}

\newaliascnt{claim}{theorem}

\aliascntresetthe{claim}

\newtheorem*{claim*}{Claim}

\newaliascnt{fact}{theorem}

\aliascntresetthe{fact}

\newtheorem*{fact*}{Fact}

\newaliascnt{observation}{theorem}

\aliascntresetthe{observation}

\newtheorem*{observation*}{Observation}

\newaliascnt{conjecture}{theorem}
\newtheorem{conjecture}[conjecture]{Conjecture}
\aliascntresetthe{conjecture}

\newtheorem*{conjecture*}{Conjecture}

\newaliascnt{corollary}{theorem}
\newtheorem{corollary}[corollary]{Corollary}
\aliascntresetthe{corollary}

\newtheorem*{corollary*}{Corollary}

\newaliascnt{remark}{theorem}

\aliascntresetthe{remark}

\newtheorem*{remark*}{Remark}

\newaliascnt{proposition}{theorem}

\aliascntresetthe{proposition}

\newtheorem*{proposition*}{Proposition}

\makeatletter
\patchcmd{\ALG@step}{\addtocounter{ALG@line}{1}}{\refstepcounter{ALG@line}}{}{}
\newcommand{\ALG@lineautorefname}{Line}
\makeatother

\usepackage{fancybox}
\usepackage{tikz}
\usetikzlibrary{shapes,backgrounds}
\usepackage{tikzsymbols}

\usetikzlibrary{decorations.pathmorphing}
\usetikzlibrary{decorations.pathreplacing}

\tikzset{snake it/.style={decorate, decoration=snake}}

\newcommand{\Pisuc}{\protocol_{\sf succinct}}
\newcommand{\LH}{\mathsf{LH}}
\newcommand{\QPCP}{\mathsf{QPCP}}

\newcommand{\lmin}{\lambda_{\sf min}}

\newcommand{\Tr}{\mathrm{Tr}}

\newcommand{\spz}[1]{|#1\rangle}

\newcommand{\rpz}[1]{\langle #1 |}

\newcommand{\yes}{\mathsf{yes}}
\newcommand{\no}{\mathsf{no}}

\newcommand{\ayes}{L_{\sf yes}}
\newcommand{\ano}{L_{\sf no}}

\newcommand{\protocol}{\Pi}

\newcommand{\regfont}{\mathbf}
\newcommand{\state}{\regfont{state}}

\newcommand{\com}{\regfont{com}}

\newcommand{\data}{\regfont{data}}

\newcommand{\proj}[1]{\spz{#1}\rpz{#1}}

\newcommand{\reg}{\regfont{reg}}

\newcommand{\caV}{\mathcal{V}}

\newcolumntype{C}[1]{%
 >{\vbox to 4ex\bgroup\vfill\centering}%
 p{#1}%
 <{\egroup}} 

\def\ShowAuthNotes{1}
\ifnum\ShowAuthNotes=1
\newcommand{\authnote}[2]{\ \\ \textcolor{red}{\parbox{0.9\linewidth}{[{\footnotesize {\bf #1:} { {#2}}}]}}\newline}
\else
\newcommand{\authnote}[2]{}
\fi

\ifnum\ShowAuthNotes=1
\newcommand{\ramis}[1]{\textcolor{magenta}{[RM: \emph{#1}]}}
\else
\newcommand{\ramis}[1]{}
\fi

\title{Quantum Merkle Trees} % Adds your title
\date{} % Adds the current date to your “cover” page; leave empty if you do not want to add a date
\author{Lijie Chen}
\affiliation{Miller Institute for Basic Research in Science, University of California, Berkeley, CA, 94720, USA}
\email{lijiechen@berkeley.edu}
\author{Ramis Movassagh}
\affiliation{Google Quantum AI, Venice, CA, 90291, USA\\
IBM Quantum, MIT-IBM Watson AI Lab, Cambridge, MA, 02142, USA}
\email{q.eigenman@gmail.com}

\newcommand{\QHROM}{\mathsf{QHROM}}
\newcommand{\Haar}{\mathbb{U}}
\newcommand{\getsR}{\in_{\sf R}}

\newcommand{\QROM}{\mathsf{QROM}}

\newcommand{\commit}{\mathsf{commit}}
\newcommand{\decommit}{\mathsf{decommit}}

\newcommand{\LocalQMA}{\mathsf{LocalQMA}}

\newcommand{\Had}{\mathsf{H}}

\begin{document}
	
	\maketitle

\newcommand{\blkpara}{b}

\begin{abstract}
	Committing to information is a central task in cryptography, where a party (typically called a prover) stores a piece of information (\eg,~a bit string) with the promise of not changing it. This information can be accessed by another party (typically called the verifier), who can later learn the information and verify that it was not meddled with. Merkle trees~\cite{Merkle87} are a well-known construction for doing so in a succinct manner, in which the verifier can learn any part of the information by receiving a short proof from the honest prover. Despite its significance in classical cryptography, there was no quantum analog of the Merkle tree. A direct generalization using the Quantum Random Oracle Model ($\QROM$)~\cite{BonehDFLSZ11} does not seem to be secure. 
	In this work, we propose the \emph{quantum Merkle tree}. It is based on what we call the \emph{Quantum Haar Random Oracle Model} ($\QHROM$). In $\QHROM$, both the prover and the verifier have access to a \emph{Haar} random quantum oracle $\caG$ and its inverse.
	
	Using the quantum Merkle tree, we propose a succinct quantum argument for the Gap-$k$-Local-Hamiltonian problem. Assuming the Quantum PCP conjecture is true, this succinct argument extends to all of $\QMA$. This work raises a number of interesting open research problems.  
	
	%We prove it is secure against semi-honest provers in $\QHROM$ 

\end{abstract}

\section{Introduction}

A commitment scheme~\cite{BrassardCC88} is a cryptographic primitive that allows a party (\ie, a prover) to (1) commit to a piece of information such as a bit string while keeping it hidden from others and (2) reveal the information they have committed to later. Commitment schemes are designed to ensure that a party cannot change the information after they have committed to it. Commitment schemes have numerous applications in cryptography, such as the construction of protocols for secure coin flipping, zero-knowledge proofs, and secure computation.

The Merkle tree~\cite{Merkle87} is an efficient example of commitment schemes, which captures the following scenario: There are two parties, the prover $\caP$ and the verifier $\caV$. $\caP$ first computes a short string called the commitment which is denoted by $\commit(x)$ from a long input string $x$ and sends $\commit(x)$ to $\caV$. Then $\caV$ asks $\caP$ to reveal a subset of bits of $x$ together with a short message that would enable $\caV$ to verify that the string $x$ has not been altered. The security promise is that after $\caP$ sends $\commit(x)$ to $\caV$, then upon $\caV$'s request of any subset of bits, a computational bounded $\caP$ can only reveal those bits faithfully.  Namely, if $\caP$ claims that the $i$-th bit of $x$ is the wrong value $1-x_i$, then her claim will be rejected by $\caV$ with high probability.

The Merkle tree has wide applications in cryptography since it allows $\caP$ to \emph{delegate} a potentially very long string to $\caV$ (\ie, a database) while enabling $\caV$ to maintain an \emph{efficient verifiable} random access to that string (say to any subset of the bits of the string). A well-known application of the Merkle tree is the succinct arguments for $\NP$ from probabilistically checkable proofs~\cite{Kilian92,Micali00} or interactive oracle proofs~\cite{Ben-SassonCS16}, where by succinctness one means that the total communication between the prover and verifier constitutes a small number of bits, say $\polylog(n)$ bits of communication.

Despite being very influential in (classical) cryptography, there is no known quantum analog of the Merkle tree that allows committing to quantum states. Such a quantum analog is appealing since it would allow a party to commit to a large quantum state $\sigma$ while maintaining verifiable access to individual qubits.

Protocols based on the classical Merkle tree are often analyzed in the random oracle model. There are also quantum models such as the Quantum Random Oracle Model~\cite{BonehDFLSZ11} ($\QROM$) for analyzing the quantum attacks against the classical Merkle tree. There are works showing that classical Merkle-tree-based protocols are secure against quantum attacks~\cite{ChiesaMS19,Chiesa21}. These works showed that commitment to classical bit strings by the Merkle tree cannot be broken by quantum adversaries. Here we hope to obtain a quantum analog of the Merkle tree that can be used to commit to quantum states.

In this work, we propose a new random oracle model which we call the Quantum Haar Random Oracle Model ($\QHROM$). We then use it in our construction of the {\it Quantum Merkle tree}. We then use it to propose a quantum analog of Kilian's succinct argument for $\NP$ and conjecture its security.

\subsection{The Merkle Tree Algorithm} Our definition of $\QHROM$ is motivated by our adaptation of the Merkle tree to the quantum setting, so it is instructive to recall the standard Merkle tree algorithm.

Let $\blkpara \in \N$ be the block-length parameter. We assume that both $\caP$ and $\caV$ have access to a random oracle function $h\colon \bits^{2\blkpara} \to \bits^{\blkpara}$. For simplicity of the argument, we will first focus on the simplest non-trivial case of a Merkle tree with two leaves and depth one, and take $n = 2\blkpara$ to be the length of the string that $\caP$ wishes to commit to. Here the string $x$ resides on the leaves and $\commit(x)$ string resides on the root. As we will see shortly, a straightforward adaption of the Merkle tree to the quantum setting is not secure even in this simple setting.

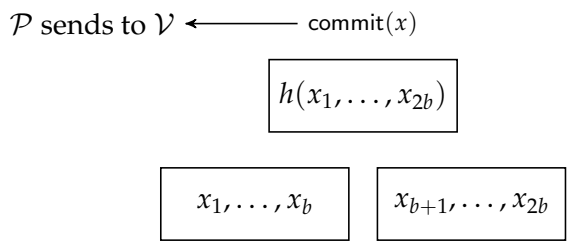
\begin{figure}[H]
	\centering
	\begin{tikzpicture}[->,>=stealth',shorten >=1pt,auto,
	semithick,scale = 1.2]
	\tikzstyle{every state}=[text=black,rectangle, minimum  width=2.5cm]
	\tikzstyle{arrow}=[thick]
	
	\node (comname) at (-0.1,2) {\footnotesize $\commit(x)$};
	\node [state] (com) at (-0.1,1.2) {$h(x_1,\dotsc,x_{2\blkpara})$};
	\node [state] (data1) at (-1.3,0) {$x_1,\dotsc,x_\blkpara$};
	\node [state] (data2) at (1.1,0) {$x_{\blkpara+1},\dotsc,x_{2\blkpara}$};
	
	\draw (comname) -- (-2.1,2);
	\node (send) at (-3.1,2) {$\caP$ sends to $\caV$};
	
	\end{tikzpicture}
	
	\caption{An illustration of the toy example for the classical Merkle tree}
	\label{fig:toy-example-classic}
\end{figure}

In this simplified setting, the protocol starts by $\caP$ simply sending the hash value $\tilde{h} = h(x)$ of $x \in \bits^{2\blkpara}$ as the $\commit(x)$ of length $\blkpara$ to $\caV$ (see~\autoref{fig:toy-example-classic} for an illustration). Then $\caV$ requests the values of a subset of bits in $x$, for which the honest $\caP$ simply responds by revealing the whole string $x$ to $\caV$. Then $\caV$ checks that the string has the same hash value $\tilde{h}$. If a (dishonest) $\caP$ can first commit to $x$ and later convince $\caV$ that its $i$-th bit is $1 - x_i$, then $\caP$ has found two strings $x \ne \tilde{x}$ with $h(x) = h(\tilde{x})$. This requires at least $2^{\blkpara/2}$ queries to the random oracle $h$ due to the birthday paradox, which is infeasible.

\subsection{A Failed Attempt to Adapt Merkle Tree in the Quantum Setting} Let us see how one might directly try to adapt the special case above of the Merkle tree algorithm to the quantum setting. An immediate idea is that, given a $2\blkpara$-qubit quantum state $\spz{\psi} = \sum_{z} \alpha_z \spz{z}$ in the register denoted by $\data$, $\caP$ treats $h$ as a quantum oracle $O_h$\footnote{That is, $O_h \spz{x}\spz{y} = \spz{x}\spz{y \oplus h(x)}$, where $x\in \bits^{2\blkpara}$, $y \in \bits^\blkpara$, and $\oplus$ denotes the entry-wise addition over $\GF(2)$.}, creates $\blkpara$ qubits initialized to $\spz{0^\blkpara}$ in register $\com$, applies $O_h$ to both $\data$ and $\com$ to obtain $\sum_{z} \alpha_z \spz{z} \spz{h(z)}$, and sends the $\com$ register to $\caV$; see~\autoref{fig:toy-example} for an illustration.

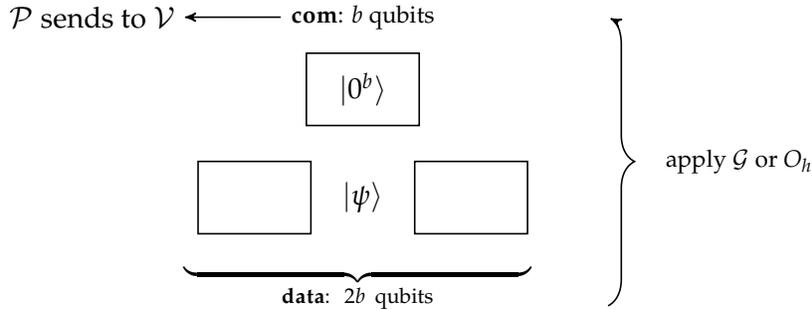
\begin{figure}[H]
	\centering
	\begin{tikzpicture}[->,>=stealth',shorten >=1pt,auto,
	semithick,scale = 1.2]
	\tikzstyle{every state}=[text=black,rectangle, minimum  width=1.5cm]
	\tikzstyle{arrow}=[thick]
	
	\node (comname) at (-0.1,2) {\footnotesize $\com$: $\blkpara$ qubits};
	\node [state] (com) at (-0.1,1.2) {$\spz{0^\blkpara}$};
	\node [state] (data1) at (-1.3,0) {};
	\node [state] (data2) at (1.1,0) {};
	\node (phistate) at (-0.1,0) {$\spz{\psi}$};
	
	\node[text width=5cm] (applyG) at (0,-1) {$\underbrace{~~~~~~~~~~~~~~~~~~~~~~~~~~~~~~~~~~~~}_{\text{$\data$: $2\blkpara$ qubits}}$};
	
	\draw [decorate,decoration={brace,amplitude=10pt,mirror,raise=4pt},yshift=0pt]
	(2.5,-1.2) -- (2.5,2) node [black,midway,xshift=3.0cm] {\footnotesize apply $\caG$ or $O_h$};
	
	\draw (comname) -- (-2.1,2);
	\node (send) at (-3.1,2) {$\caP$ sends to $\caV$};
	
	\end{tikzpicture}
	
	\caption{An illustration of the toy example for the quantum Merkle tree}
	\label{fig:toy-example}
\end{figure}

To reveal qubits in $\spz{\psi}$, $\caP$ simply sends the $\data$ register to $\caV$ as well, and $\caV$ applies $O_h$ again to both $\data$ and $\com$, and measures $\com$ in the computational basis to check if it is $0^{\blkpara}$ and rejects immediately otherwise. However, this is not secure against \emph{phase attack}. After sending $\com$ to $\caV$, for every Boolean function $f\colon \bits^{2\blkpara} \to \bits$, $\caP$ can apply the unitary $\spz{z} \mapsto (-1)^{f(z)} \spz{z}$ to $\data$, and then sends it to $\caV$. One can see that $\caV$ still accepts this state with probability $1$, but $\caP$ has cheated by changing the state from $\sum_{z} \alpha_z \spz{z}$ to $\sum_{z} (-1)^{f(z)} \alpha_z \spz{z}$, which can be entirely a different state for some function $f$.

The issue above is that the mapping $O_h$, $\spz{x}\spz{y} \mapsto \spz{x}\spz{y \oplus h(x)}$, has too much structure to be exploited by the attacker. This immediately suggests to us to consider a more random choice of quantum oracles which indeed we take to be \emph{the most} random choice of quantum oracles: a Haar random quantum oracle. 

Comment: One way to address the phase attack above is to make $O_h$ more complicated. For example, instead of applying $O_h$ once to the registers $\data$ and $\com$, we can repeatedly apply $O_h \Had^{\otimes 2\blkpara}$ several times ($\Had$ denotes the Hadamard gate). We found such a construction more cumbersome and even harder to analyze compared to a Haar random unitary. Moreover, it is conjectured~\cite[Section~6]{JiL018} that similar constructions may already be indistinguishable from a Haar random unitary (see~\autoref{sec:suc-discussions} for more discussions). Hence, it seems more natural to directly work with a Haar random unitary.

\subsection{The Quantum Haar Random Oracle Model ($\QHROM$) and Quantum Merkle Tree} 
We are now ready to introduce the Quantum Haar Random Oracle Model ($\QHROM$). In $\QHROM$ both $\caP$ and $\caV$ have access to a Haar random quantum oracle $\caG$ and its inverse $\caG^\dagger$ that act on $3\blkpara$ qubits (see~\autoref{defi:QHROM} for a precise definition). The protocol between $\caP$ and $\caV$ remains the same for the special case $n = 2\blkpara$ except for replacing $O_h$ by $\caG$. It is easy to see that since $\caG$ completely \emph{obfuscates} the original state $\spz{\psi}$, the phase attack described above no longer applies.

Next, we describe the quantum Merkle tree in the general setting in which $n$ can be arbitrarily large and denotes the number of qubits in the state that $\caP$ wishes to commit to. Given a quantum state $\sigma$ on $n = \blkpara \cdot \ell$ qubits for some $\ell=2^d$ and $d \in \N$,\footnote{We can always pad any quantum state to such length by adding dummy qubits. This at most doubles the number of qubits.} we partition $x$ into $\ell$ consecutive blocks of length $\blkpara$ as $x^{(1)},x^{(2)},\dotsc,x^{(\ell)}$. Then, we build a perfect binary tree with $\ell$ leaves (see~\autoref{fig:binarytree}), where each leaf corresponds to a block of the input. Next, from the leaves to the root, we assign to each node $\alpha$ a $\blkpara$-qubit register $\com_\alpha$ as follows: (1) if $\alpha$ is a leaf, then $\com_\alpha$ is simply the qubits of the assigned block and (2) if $\alpha$ is an intermediate node with two children $\beta$ and $\gamma$, then we initialize $\com_\alpha$ to $\spz{0^\blkpara}$, and apply $\caG$ to the three registers $\com_\beta$, $\com_\gamma$, and $\com_\alpha$. Finally, $\caP$ sends the register $\com_{\sf rt}$ to $\caV$, where ${\sf rt}$ is the root of the binary tree.

Suppose $\caV$ requests the state of the $i$-th block $x^{(i)}$ of the quantum state. To reveal the $i$-th block $x^{(i)}$ on a leaf (which we denote  $\mu$) of the tree $\caP$ sends all the $\com_\alpha$ for nodes $\alpha$ that are the (1) ancestors of $\mu$, (2) siblings of an ancestor of $\mu$, or (3) $\mu$ or the sibling of $\mu$. $\caV$ then ``undoes'' all the applied $\caG$ in the exact reverse order by applying $\caG^\dagger$ to the registers sent by $\caP$ starting from the register $\com_{\sf rt}$, and then from the root downwards to the leaves. After that, for every ancestor $\alpha$ of $\mu$, $\caV$ checks that $\com_\alpha$ is $\spz{0^\blkpara}$ by measuring it in the computational basis.   To illustrate, if $\caV$ asks for the block $x^{(2)}$, then $\caP$ sends the corresponding $\com_\alpha$ registers for all diamond shape nodes in~\autoref{fig:binarytree-intro}.

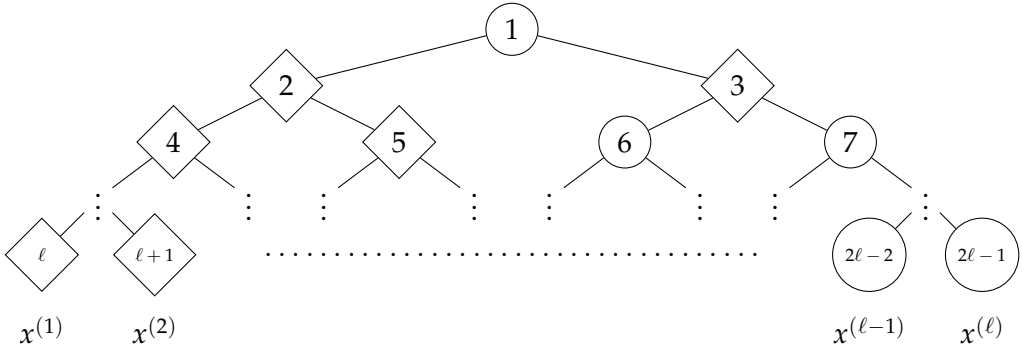
\begin{figure}
	\begin{center}
		\begin{tikzpicture}[level/.style={sibling distance=120mm/#1},scale = 0.5]
		\node [circle,draw]{$1$}
		child {
			node [diamond,draw] {$2$}
			child {
				node [diamond,draw] {$4$}
				child {
					node {$\vdots$} 
					child {node [diamond,draw,scale = 0.75,  minimum size=1.3cm] (aa) {\footnotesize$\ell$}}
					child {node [diamond,draw,scale = 0.75,  minimum size=1.3cm] (a) {\footnotesize$\ell + 1$}}
				}
				child {node {$\vdots$}}
			}
			child {
				node [diamond,draw] {$5$}
				child {node {$\vdots$}}
				child {node {$\vdots$}}
			}
		}
		child {
			node [diamond,draw] {$3$}
			child {
				node [circle,draw] {$6$}
				child {node {$\vdots$}}
				child {node {$\vdots$}}
			}
			child {
				node [circle,draw] {$7$}
				child {node {$\vdots$}}
				child {node {$\vdots$}
					child {node [circle,draw,scale = 0.75,  minimum size=1.3cm] (b) {\footnotesize$2\ell - 2$}}
					child {node [circle,draw,scale = 0.75,  minimum size=1.3cm] (bb){\footnotesize$2\ell - 1$}}
				}
			}
		};
		\path (a) -- (b) node [midway] {$\cdots\cdots\cdots\cdots\cdots\cdots\cdots\cdots\cdots\cdots\cdots\cdots$};
		\node [below of= aa] {$x^{(1)}$};
		\node [below of= a] {$x^{(2)}$};
		\node [below of= b] {$x^{(\ell-1)}$};
		\node [below of= bb] {$x^{(\ell)}$};
		\end{tikzpicture}
		\caption{An illustration of the quantum Merkle tree with $\ell = 2^{d}$ input blocks; when block $x^{(2)}$ is requested by $\caV$, $\caP$ sends all the diamond shape nodes.}\label{fig:binarytree-intro}
	\end{center}
\end{figure}

Comment: So how might one heuristically instantiate a Haar-Random unitary? One might use a random quantum circuit that well approximates the behavior of a Haar unitary. For example, one might use a polynomially deep circuit. One way to formalize the degree of approximation is via the ideas in $k$-design~\cite{brandao2016local}. 

\subsection{A Candidate for Succinct Quantum Argument for Gap-$k$-$\LH$ in $\QHROM$}

Similar to Kilian's succinct argument for $\NP$, the quantum Merkle tree naturally suggests a succinct argument $\Pisuc$ for the Gap Local Hamiltonian Problem. We first recall its definition below.

\begin{definition}
	(Gap-$k$-Local Hamiltonian Problem) Given  $\alpha,\beta$ with $0<\alpha<\beta\le 1$ and a $k$-local Hamiltonian with $m$ local terms $\{ H_i \}_{i \in [m]}$ such that $0 \le  H_i \le I$, decide whether $\lmin(\sum_{i=1}^{m}H_i)$ is at most $\alpha m$ or at least $\beta m$. Below we abbreviate this problem by $(\alpha,\beta)\text{-}k\text{-}\LH$.
\end{definition}

Formally, in $\Pisuc$ the honest prover $\caP$ applies the quantum Merkle tree to a ground state $\sigma$ of $\sum_{i=1}^{m} H_i$, and sends $\com_{\sf rt}$ to $\caV$. Then $\caV$ draws an integer $i$ from $\{1,2,\dotsc,m\}$ uniformly at random and asks $\caP$ to reveal the qubits in the support of the term $H_i$.  $\caV$ does the decommitment from the root towards the qubits in the support of $H_i$ as described above. If in this decommitment phase the ancestors of the qubits in the support of $H_i$ all result in $0^\blkpara$ it proceeds to the last step. In the last step, it measures the POVM $\{H_i,1 - H_i\}$ on the qubits in the support of $H_i$ and rejects if it sees $H_i$. Indeed, this is the natural analog of Kilian's succinct argument~\cite{Kilian92} in the quantum setting.

We prove that if $\caP$ follows the protocol, then (1) when $\lambda_{\sf min}(\sum_{i=1}^m H_i) \le \alpha \cdot m$, $\caP$ can make $\caV$ accept with probability at least $1 - \alpha$, and (2) when $\lambda_{\sf min}(\sum_{i=1}^m H_i) \ge \beta \cdot m$, $\caP$ cannot force $\caV$ to accept with a probability greater than $1 - \beta < 1 - \alpha$ (See~\autoref{theo:Pisuc-c-and-s} for details). By a sequential repetition argument, the completeness $1-\alpha$ and the soundness $1-\beta$ can be boosted to $1 - n^{-\omega(1)}$ and $n^{-\omega(1)}$ respectively where $\omega (1)$ means super constant.

However, a malicious $\caP$ may not follow the protocol, but instead come up with some arbitrary states for the different nodes that are not the result of the quantum Merkle tree algorithm and send those to $\caV$ instead. We currently do not know how to analyze such an arbitrary attack, but we conjecture the following:
\begin{conjecture}\label{conj:main-conj}
	For the constants $k \in \N$ and $0 < \alpha < \beta \le 1$, $\Pisuc$ (with sequential repetition) for $(\alpha,\beta)$-$k$-$\LH$ has completeness $1 - n^{-\omega(1)}$ and soundness $n^{-\omega(1)}$ in $\QHROM$.
\end{conjecture}

\subsection{Open Questions and Follow-up Works}

We believe our inability to prove~\autoref{conj:main-conj} is mainly due to the lack of tools available for analyzing this new $\QHROM$ setting. We remark that only two years ago~\cite{ChiesaMS19} managed to prove that the succinct argument for $\NP$~\cite{Kilian92,Micali00} is secure in the $\QROM$ model by using the recently proposed \emph{compressed oracles} technique introduced in~\cite{Zhandry19} which gives a nice way to analyze $\QROM$. To prove the security of our succinct argument for Gap-$k$-$\LH$, one likely needs similar advances for analyzing the $\QHROM$. We now list some specific open problems:

\begin{open}
	Is there an analog of the compressed oracle technique in~\cite{Zhandry19} for the $\QHROM$?
\end{open}
Above we generalized Kilian's constant-round succinct argument~\cite{Kilian92} to the quantum setting and conjectured its soundness. A natural open question is whether we can generalize Micali's non-interactive succinct argument for $\NP$~\cite{Micali00} to the quantum settings as well.

\begin{open}
	Is there an analog of Micali's non-interactive succinct argument for Gap-$k$-$\LH$?
\end{open}

A particularly useful feature of previous succinct arguments for $\NP$~\cite{Kilian92,Micali00} is that they can be made \emph{zero-knowledge} with minimal overhead. A natural open question is whether we can make our proposed succinct argument for Gap-$k$-$\LH$ zero-knowledge as well.

\begin{open}
	Is there a zero-knowledge succinct argument for Gap-$k$-$\LH$ in $\QHROM$?
\end{open}

%Another open question is whether we can instantiate the protocol $\Pisuc$ in the more standard $\QROM$. A natural way to do this is to simulate a Haar random unitary by a standard random oracle; see~\cite[Section~6]{JiL018} for some candidate constructions. 

%\begin{open}
%	Is there a succinct argument for Gap-$k$-$\LH$ in $\QROM$?
%\end{open}

\paragraph*{Subsequent work.} Finally, we remark that this paper formed the basis of the ideas in an exciting subsequent work~\cite{gunn2022commitments}. They proved the security of a tree commitment similar to what is presented here but from standard (quantum) cryptographic assumptions. Note that it is not a priori clear what "security" even means for commitments to quantum states\footnote{Although such commitment scheme has appeared implicitly in several quantum cryptographic protocols~\cite{BroadbentJSW20,ColadangeloVZ20,BroadbentG22}.}, and a major contribution of~\cite{gunn2022commitments} is to formally define the security of commitments to quantum states. We refer readers to~\cite{gunn2022commitments} for an overview of more prior works on quantum commitment schemes. As far as we know, the security of the precise protocol (in the $\QHROM$) given in this paper remains open.
	
\section{Preliminaries}\label{sec:prelim}

\subsection{Notation}

We always denote by $\log$ the logarithm in base $2$. We denote by $[n]$ the set of integers $\{1,2,\dots,n\}$.  Let $\reg$ be a register of $n$ qubits. For each $i \in [n]$, $\reg(i)$ denotes the $i$-th qubit in $\reg$, and $\reg[\ell,r]$ denotes the qubits from $\reg(\ell)$ to $\reg(r)$. The corresponding Hilbert space is denoted by $\caH_{\reg}$. For $k$ pairwise-disjoint sets $S_1,\dotsc,S_k$, we use $\bigsqcup_{i \in [k]} S_i$ to denote their union. We say a function $\alpha \colon \N \to [0,1]$ satisfies $\alpha(n) \le \negl(n)$ (\ie, $\alpha$ is \emph{negligible}), if for all constant $k \ge 1$, $\lim_{n \to \infty} \alpha(n) \cdot n^k = 0$ (\ie, $\alpha(n) = o(n^{-k})$ for every $k \in \N$).

For a quantum state $\sigma$ on $n$ qubits and a subset $S \subseteq [n]$, $\sigma_{S} \coloneqq \Tr_{[n] \setminus S}[\sigma]$ is the reduced density matrix. For a quantum state $\spz{\psi} \in \caH_{\reg}$, for simplicity we sometimes use $\psi$ to denote the corresponding density matrix $\psi=\proj{\psi}$. Given a unitary sequence $U_{1},\dotsc,U_{T}$, we write $U_{[\ell,r]}$ to denote the product $U_{r}U_{r-1}\cdots U_{\ell}$ for ease of notation.

For two quantum states $\sigma$ and $\rho$, we use $\| \sigma - \rho \|_1$ to denote their trace distance. We also write $x \getsR A$ to mean that $x$ is drawn from the set $A$ uniformly at random.

\subsection{The Quantum Haar Random Oracle Model}

We will consider the \emph{Quantum Haar random oracle model} ($\QHROM$), in which every agent (prover and verifier) gets access to a Haar random oracle $\caG$ acting on $\lambda$ qubits and its inverse $\caG^\dagger$, where $\lambda$ is the so-called \emph{security parameter}.

We denote by $\Haar(N)$ the set of all $N \times N$ unitaries. By $\caG \getsR \Haar(N)$ we mean that $\caG$ is an $N \times N$ unitary drawn from the Haar measure.

\begin{definition}\label{defi:QHROM}
	An interactive protocol $\Pi$ between the prover $\caP$ and verifier $\caV$ is a proof system for a promise problem $L=(\ayes,\ano)$ with completeness $c(n,\lambda)$ and soundness $s(n,t,\lambda)$ in $\QHROM$, if the following holds:
	\begin{description}
		\item[] \textbf{$\caP$ and $\caV$:} $\caP$ and $\caV$ are both given an input $x \in \ayes \cup \ano$. $\caV$ is polynomial-time and outputs a classical bit indicating acceptance or	rejection of $x$, and $\caP$ is unbounded. Both $\caV$ and $\caP$ are given access to a Haar random quantum oracle $\caG$ and its inverse $\caG^\dagger$ that act on $\lambda$ qubits (that is, $\caG \getsR \Haar(2^\lambda)$). Let $n = |x|$. 
		\item[] \textbf{Completeness:}
		{ If $x\in\ayes$, 
			\[
			\Ex_{\caG \getsR \Haar(2^\lambda)} \Pr[(\caV^{\caG,\caG^\dagger} \leftrightarrows \caP^{\caG,\caG^\dagger} )(x) = 1] \geq c(n,\lambda).
			\]
		}
		\item[] \textbf{Soundness:} {If $x\in\ano$, for every $t \in \N$ and  any unbounded prover $\caP^*$ making at most $t$ total queries to $\caG$ and $\caG^\dagger$, we have that
			\[
			\Ex_{\caG \getsR \Haar(2^\lambda)} [(\caV^{\caG,\caG^\dagger}  \leftrightarrows (\caP^*)^{\caG,\caG^\dagger} )(x) = 1)] \leq s(n,t,\lambda).
			\]
		}
	\end{description}
	In the above $\leftrightarrows$" denotes the interactive nature of the protocol between $\caP$ and $\caV$.
\end{definition}

We remark that in the soundness part, the only restriction on a malicious prover $\caP^*$ is the number of queries it can make to $\caG$ and $\caG^\dagger$. In particular, this means that even if $\caP^*$ has unbounded computational power, as long as it makes a small number of queries to $\caG$ and $\caG^\dagger$, it cannot fool the verifier.

\subsection{Quantum Local Proofs}

Next, we provide formal definitions of $\LocalQMA$. For a reader familar with $\QMA$ in the following definition, one can think of $x$ as the classical description of the local Hamiltonian problem, and $m(n)$ as the number of terms in it (\ie,~$H=\sum^m_{i=1} H_i$).

\begin{definition}[$(k,\gamma)$-$\LocalQMA$]\label{defi:localQMA}
	For two constants $k,\gamma \in \N$, a promise problem $L = (L_{\yes}, L_{\no})$ is in the complexity class $(k,\gamma)$-$\LocalQMA$ with soundness $s(n)$ and completeness $c(n)$ if there are polynomials $m$ and $p$ such that the following hold:
	
	\begin{itemize}
		\item (\textbf{A $k$-local verifier $V_L$}) Let $n = |x|$. There is a verifier $V_L$ that acts as follows:
		\begin{enumerate}
			\item $V_L$ gets access to a $p(n)$-qubit proof $\sigma$ for $L$ and draws $i \getsR [m(n)]$. $V_L$ then computes in $\poly(n,k,\gamma)$ time a $k$-size subset $S_i \subseteq [p(n)]$ and a $\gamma$-size quantum circuit $C_i$ that is over the Clifford + T gate-set and acts on $k$ qubits. $C_i$ may use $\gamma$ ancilla qubits, with the first ancilla qubit being the output qubit.
			
			\item $V_L$ next applies $C_i$ to the restriction of $\sigma$ on qubits in $S_i$ and measures the first ancilla qubit. $V_L$ accepts if the outcome is $1$ and rejects otherwise.
		\end{enumerate}
		
		\item (\textbf{Completeness}) If $x \in L_{\yes}$, there is a $p(n)$-qubit state $\sigma$ such that $V_L$ accepts $\sigma$ with probability at least $c(n)$.			
		
		\item (\textbf{Soundness}) If $x \in L_{\no}$, $V_L$ accepts every $p(n)$-qubit state $\sigma$ with probability at most $s(n)$.
		
		\item (\textbf{Strongly explicit}) Moreover, we say that $V_L$ is strongly explicit, if $V_L$ computes $S_i$ and $C_i$ in $\poly(\log n,k,\gamma)$ time instead of $\poly(n,k,\gamma)$ time.
	\end{itemize}
\end{definition}

We will use $(k,\gamma)$-$\LocalQMA_{s,c}$ to denote the class above for notational convenience.

\subsection{The Quantum $\PCP$ Conjecture}\label{sec:QPCP}

We first recall the quantum $\PCP$ conjecture~\cite{AharonovALV09,AharonovAV13}.

\begin{conjecture}[$\QPCP$ conjecture]
	There are constants $k \in \N$ and $\alpha,\beta$ satisfying $0 < \alpha < \beta \le 1$ such that $(\alpha,\beta)$-$k$-$\LH$ is $\QMA$-complete.
\end{conjecture}

In particular, the following corollary is immediate from the definition of $(\alpha,\beta)$-$k$-$\LH$.

\begin{corollary}\label{cor:QPCP-and-LocalQMA}
	If $\QPCP$ holds, then there are constants $k,\gamma \in \N$ and $c,s \in [0,1]$ satisfying that $s < c$, such that
	\[
	\QMA \subseteq (k,\gamma)\text{-}\LocalQMA_{s,c}.
	\]
\end{corollary}
	
\section{A Candidate Succinct Argument for $\LocalQMA$ in $\QHROM$}\label{sec:succinct}

In this section, we present a candidate succinct argument for $\LocalQMA$ in $\QHROM$. Assuming $\QPCP$, this succinct argument also works for all of $\QMA$.

\subsection{The Succinct Protocol $\Pisuc$}
\newcommand{\id}{\mathsf{id}}
\newcommand{\blk}{b}

\paragraph{Notation.} 

Let $L = (L_{\yes}, L_{\no}) \in (k,\gamma)\text{-}\LocalQMA_{s,c}$ for two integers $k,\gamma \in \N$ and two reals $s,c \in [0,1]$ such that $s < c$. Let $m_L$ and $p_L$ be the polynomials and $V_L$ be the $k$-local verifier in~\autoref{defi:localQMA}. Throughout this section, we will always use $n$ to denote the length of the input to $L$, $N = p_L(n)$ to denote the number of qubits in a witness for $V_L$, and $\lambda$ to denote the security parameter. 

We set $\blk = \lambda /3$, and $\ell = N/\blk$. We assume that $\blk$ is an integer and $\ell$ is a power of $2$ for simplicity and without loss of generality since one can always add dummy qubits to the witness.

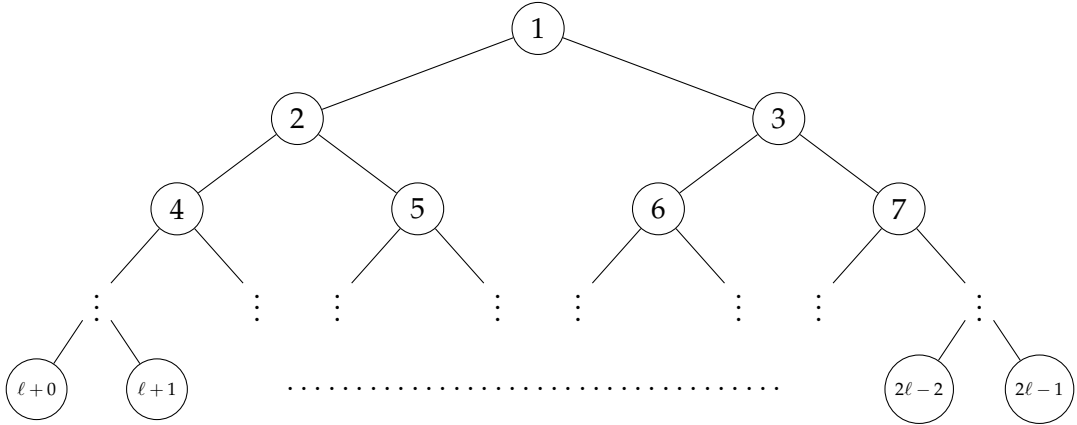
\begin{figure}[H]
\begin{center}
\begin{tikzpicture}[level/.style={sibling distance=80mm/#1},scale = 0.8]
    \node [circle,draw]{$1$}
      child {
      	node [circle,draw] {$2$}
      	child {
      		node [circle,draw] {$4$}
      		child {
      			node {$\vdots$} 
      			child {node [circle,draw,scale = 0.75] {\footnotesize$\ell+0$}}
      			child {node [circle,draw,scale = 0.75] (a) {\footnotesize$\ell+1$}}
      		}
      		child {node {$\vdots$}}
      	}
      	child {
      		node [circle,draw] {$5$}
      		child {node {$\vdots$}}
      		child {node {$\vdots$}}
      	}
      }
      child {
      	node [circle,draw] {$3$}
      	child {
      		node [circle,draw] {$6$}
      		child {node {$\vdots$}}
      		child {node {$\vdots$}}
      	}
      	child {
      		node [circle,draw] {$7$}
      		child {node {$\vdots$}}
      		child {node {$\vdots$}
      			child {node [circle,draw,scale = 0.75] (b) {\footnotesize$2\ell-2$}}
      			child {node [circle,draw,scale = 0.75] {\footnotesize$2\ell-1$}}
      		}
      	}
      };
      \path (a) -- (b) node [midway] {$\cdots\cdots\cdots\cdots\cdots\cdots\cdots\cdots\cdots\cdots\cdots\cdots$};
\end{tikzpicture}
\caption{An illustration of the labeling of the nodes in the tree $T_\ell$ with $\ell$ leaves}\label{fig:binarytree}
\end{center}
\end{figure}

\paragraph*{The perfect binary tree $T_\ell$.} We will consider a perfect binary tree $T_\ell$ of $\ell$ leafs (see~\autoref{fig:binarytree} for an illustration). Note that $T_\ell$ has $\log\ell$ layers. We label the nodes in $T_\ell$ first from root to leaves and then from left to right, starting with $1$.

For a node $u$ in $T_\ell$, we observe that $u$'s parent is $\lfloor u/2 \rfloor$ if $u$ is not the root (\ie, $u \ne 1$) and $u$'s two children are $2u$ and $2u + 1$ if $u$ is not a leaf (\ie, $u < \ell$). We use $P_u$ to denote the set of nodes consisting of $u$ and all ancestors of $u$. Formally, we have
\[
P_u = \begin{cases}
\{u\} \quad& \text{if $u = 1$,}\\
\{u\} \cup P_{\lfloor u/2 \rfloor} \quad& \text{if $u > 1$.}
\end{cases}
\]

We also define $R_u$ as follows:
\[
R_u = \{\text{$v\in P_u$ or $\lfloor v/2\rfloor \in P_u$}  : v \in [2\ell - 1] \}.
\]
That is, a node $v$ belongs to $R_u$ if either $v$ is in $P_u$ or the parent of $v$ is in $P_u$. Also, for a set of nodes $S \subseteq [2\ell - 1]$, we set $R_S = \bigcup_{u \in S} R_u$.

Given an $N$-qubit state $\sigma$, we define the following commitment algorithm (\autoref{algo:commit}) and the corresponding local decommitment algorithm (\autoref{algo:decommit}).

\begin{algorithm}[H]
	\caption{Algorithm for committing to an $N$-qubit quantum state}\label{algo:commit}
	\SetKwProg{Fn}{Function}{}{}
	\Fn{$\commit^{\caG}(\sigma,N,\lambda)$}{
		\KwIn{$\sigma$ is an $N$-qubit quantum state, $\lambda$ is the security parameter (recall that $\lambda = 3\blk$)}
		Let $\ell = N/\blk$\; 
		For each node $u$ in $T_\ell$, create a $\blk$-qubit register $\state^{(u)}$\;
		Store $\sigma$ in registers $\state^{(\ell)},\state^{(\ell+1)},\dotsc,\state^{(2\ell-1)}$\;
		\For{$u$ from $\ell - 1$ down to $1$}{
			Initialize $\state^{(u)}$ as $\spz{0^\blk}$\;
			Apply $\caG$ on $\state^{(2u)}$, $\state^{(2u+1)}$, and $\state^{(u)}$\;
		}
		\Return all qubits in $\{ \state^{(u)} \}_{u \in [2\ell - 1]}$ 
            \tcp*{Here $\state^{(1)}$ is the commitment to be sent to the verifier, while the prover keeps all other states $\{ \state^{(u)} \}_{u \in \{2,\dotsc,2\ell - 1\}}$}
	}
\end{algorithm}

\begin{algorithm}[H]
	\caption{Algorithm for recovering part of the original quantum state}\label{algo:decommit}
	\SetKwProg{Fn}{Function}{}{}
	\Fn{$\decommit^{\caG^\dagger}(N,\lambda, S, \{ \eta_{u} \}_{u \in R_S})$}{
		\KwIn{$S \subseteq \{\ell,\dotsc,2\ell - 1\}$ is a subset of leafs in $T_\ell$, for each $u \in R_S$,  $\eta_u$ is a $\blk$-qubit quantum state, $\lambda$ is the security parameter. (We remark that $\{ \eta_{u} \}_{u \in (R_S \setminus \{1\})}$ are the states provided by the prover to the verifier.)}
		Let $\ell = N/\blk$\; 
		For each node $u$ in $R_S$, create a $\blk$-qubit register $\state^{(u)}$, and store $\eta_u$ at $\state^{(u)}$\;
		\For{$u \in R_S \cap [\ell - 1]$, from smallest to the largest\label{line:order-decom}}{
			Apply $\caG^\dagger$ on $\state^{(2u)}$, $\state^{(2u+1)}$, and $\state^{(u)}$\;
			Measure $\state^{(u)}$ in the computational basis to obtain an outcome $z \in \bits^{\blk}$\;
			\If{$z \ne 0^{\blk}$}{\Return $\perp$ \tcp{$\perp$ means the check fails}} 
		}
		\Return all qubits in $\{ \state^{(u)} \}_{u \in S}$\;
	}
\end{algorithm}

Finally, we are ready to specify the following candidate succinct argument for $L \in (k,\gamma)\text{-}\LocalQMA_{s,c}$.

\begin{construction}{The candidate succinct argument $\Pisuc$ for $L \in (k,\gamma)\text{-}\LocalQMA_{s,c}$}
	\begin{itemize}
		\item Both prover ($\caP$) and verifier ($\caV$) get access to a Haar random quantum unitary $\caG$ acting on $3\blk = \lambda$ qubits and its inverse $\caG^\dagger$. They also both get an input $x \in \bits^n$ to $L$. The goal for the prover is to convince the verifier that $x \in L_{\yes}$.
		
		Let $\ell = N/\blk$ and we assume that $\ell = 2^d$ for $d \in \N$.
		
		\item (\textbf{First message: $\caP \to \caV$}) The honest prover $\caP$ acts as follows: If $x \in L_{\no}$, $\caP$ aborts immediately. Otherwise, $\caP$ finds an $N$-qubit state $\sigma$ such that $V_L$ accepts with probability at least $c$, and runs $\commit(\sigma,N,\lambda)$ to obtain qubits $\{ \eta_{u} \}_{u \in [2\ell - 1]}$.
		
		$\caP$ then sends $\eta_{1}$ to $\caV$.
		
		\item (\textbf{Second message: $\caV \to \caP$}) $\caV$ now simulates the local verifier $V_L$: $\caV$ first draws $i \getsR [m_L(n)]$ and sends $i$ to $\caP$, and then computes a $k$-size subset $S_i \subseteq [N]$ and a $\gamma$-size circuit $C_i$ acting on $k$ qubits, according to~\autoref{defi:localQMA}.
		
		Let $W_i$ be the set of leaves in $T_\ell$ that contain the qubits indexed by $S_i$. That is,
		\[
		W_i = \{ \ell + \lfloor (u-1)/\blk \rfloor : u \in S_i \}.
		\]
		
		\item (\textbf{Third message: $\caP \to \caV$}) The honest prover $\caP$ sends $\{\eta_u\}_{u \in R_{W_i}, u \ne 1}$ to $\caV$. $\caV$ then runs $\decommit(N,\lambda,W_i,\{ \eta_u \}_{u \in R_{W_{i}}})$ (note that $\caV$ already has $\eta_1$). If $\decommit$ returns $\perp$, $\caV$ rejects immediately.
		
		Otherwise, $\caV$ continues the simulation of $V_L$ by running $C_i$ using $\{ \eta_{u} \}_{u \in W_i}$, and $\caV$ accepts if and only if $V_L$ accepts.
	\end{itemize}
\end{construction}

\subsection{Analysis of $\Pisuc$}

We say a prover $\caP$ is \emph{semi-honest}, if $\caP$ commits to an arbitrary $N$-qubit state (as opposed to the true ground state) $\sigma$ in the first message but indeed follows $\Pisuc$. We remark that, unlike an honest prover, a semi-honest prover may not necessarily commit to a state that makes $V_L$ accepts with probability at least $c$.\footnote{In particular, a semi-honest prover may still commit to some state even when $x \in L_{\no}$, while an honest prover would abort when $x \in L_{\no}$.}

Now we prove the completeness and succinctness of $\Pisuc$. We also show $\Pisuc$ is sound against semi-honest provers.\footnote{We remark that semi-honest prover security is more of a sanity check than a solid contribution since it is not hard to construct a trivial protocol that satisfies this semi-honest prover security.}

\begin{theorem}\label{theo:Pisuc-c-and-s}
	Let $\Pisuc$ be the protocol between $\caP$ and $\caV$ for the promise language $L \in (k,\gamma)\text{-}\LocalQMA_{s,c}$. For every $x \in \bits^n$, the following hold:
	\begin{description}
		\item[] \textbf{Completeness:}
		{ 	If $x \in L_{\sf yes}$, then for every $\caG \in \Haar(2^\lambda)$,
			\[
			\Pr[(\caV^{\caG,\caG^\dagger} \leftrightarrows \caP^{\caG,\caG^\dagger} )(x) = 1] \ge c.
			\]}
		\item[] \textbf{Soundness against semi-honest provers:}
		{ 	If $x \in L_{\sf no}$, then for every $\caG \in \Haar(2^\lambda)$ and every semi-honest prover $\caP$,
			\[
			\Pr[(\caV^{\caG,\caG^\dagger} \leftrightarrows (\caP)^{\caG,\caG^\dagger} )(x) = 1] \le s.
			\]}
		\item[]  \textbf{Succinctness:} {
			$\caP$ and $\caV$ communicate at most $O(k \cdot \lambda \cdot \log n)$ qubits in total. }
		
		\item[] \textbf{Efficiency:} {$\caV$ runs in $\poly(n,k,\gamma)$ time. If $V_L$ is strongly explicit, then $\caV$ runs in $O(k \cdot \lambda \cdot \log n + \poly(\log n,k,\gamma))$ time.
		}
	\end{description}
\end{theorem}
\begin{proof}
	We first establish the succinctness part. Examining the protocol $\Pisuc$, one can see that the first message takes $O(\lambda)$ qubits, the second messages takes $O(\log m_L(n)) = O(\log n)$ classical bits, and the third message takes $O\left( |R_{W_i}| \cdot \lambda \right)$ qubits. Note that $|R_{W_i}| \le |W_i| \cdot O(\log \ell) \le k \cdot O(\log N) \le O(k \cdot \log n)$, the total communication complexity is thus bounded by $O(k \cdot \lambda \cdot \log n )$. 
	
	For the running time of $\caV$, one can see that its running time is dominated by the running time of $\decommit(N,\lambda,W_i,\{ \eta_u \}_{u \in R_{W_i}})$ and the running time of $V_L$ computing $W_i$ and $C_i$, which are at most $O(k \cdot \lambda\cdot \log N)$ and $\poly(n,k,\gamma)$ respectively. The latter becomes ($\poly(\log n,k,\gamma)$ if $V_L$ is strongly explicit.
	
	Now we prove the completeness. Let $\caG_{(u)}$ be a $\caG$ gate applying on registers $\state^{(2u)}$, $\state^{(2u+1)}$, and $\state^{(u)}$. Then we know for the honest prover $\caP$, when $x \in L_{\yes}$, it prepares an $N$-qubit state $\sigma$ that makes $V_L$ accept with probability at least $c$, and then applies $U_{\sf com} \coloneqq \caG_{(\ell-1)} \cdot \dotsc\cdot\caG_{(1)}$ to $\sigma \otimes \proj{0}_{\state^{(1)},\dotsc,\state^{(\ell-1)}}$.
	
	Let $U_{\sf decom} \coloneqq U_{\sf com}^{\dagger} = \caG_{(1)}^\dagger\cdot \dotsc \cdot \caG_{(\ell-1)}^\dagger$. Recall that verifier $\caV$ at the end simulates the quantum circuit $C_i$ only on registers in $\{ \state^{(u)} \}_{u \in W_i}$. We now argue that $\caV$ is effectively simulating $C_i$ on
	\[
	U_{\sf decom}^{\dagger} U_{\sf com} \sigma \otimes \proj{0}_{\state^{(1)},\dotsc,\state^{(\ell-1)}} = \sigma \otimes \proj{0}_{\state^{(1)},\dotsc,\state^{(\ell-1)}}.
	\]
	
	The reason is that $\decommit(N,\lambda,W_i,\{ \eta_u \}_{u \in R_{W_i}})$ performs all gates in $U_{\sf decom}$ that reside in the lightcone of the registers $\{ \state^{(u)} \}_{u \in W_i}$ in the chronological order (see Line~\ref{line:order-decom} of \autoref{algo:decommit}). Also, since $\caP$ starts with the state $\sigma \otimes \proj{0}_{\state^{(1)},\dotsc,\state^{(\ell-1)}}$, $\decommit$ never outputs $\bot$. Therefore, $\caV$ is simulating $V_L$ faithfully on $\sigma$, meaning that it accepts with probability at least $c$.
	
	Finally, we establish the soundness against semi-honest provers. The argument above for completeness indeed established that whenever the prover commits to a state $\sigma$ in the first message and follows $\Pisuc$ (\ie, the prover is semi-honest), for every possible $\caG$,	the acceptance probability of $\caV$ equals the acceptance probability of the simulated $V_L$ on $\sigma$. Hence, when $x \in L_{\no}$, for every semi-honest prover and every possible $\caG$, the acceptance probability of $\caV$ is at most $s$.
\end{proof}

We conjecture that the soundness also holds more generally.
\begin{conjecture}[$\Pisuc$ is sound in $\QHROM$]\label{conj:soundness}
	Let $\Pisuc$ be the protocol between $\caP$ and $\caV$ for the promise language $L \in (k,\gamma)\text{-}\LocalQMA_{s,c}$. For every $x \in \bits^n$, the following hold:
	\begin{description}
		\item[] \textbf{Soundness:}
		{ 
			If $x \in L_{\no}$, then for every $t \in \N$ and all (potentially malicious) $\caP^*$ that makes at most $t$ total queries to $\caG$ and $\caG^\dagger$, for some $\delta = \delta(t,\lambda) = \poly(t)/2^{\Omega(\lambda)}$, it holds that
			\[
			\Pr_{\caG \getsR \Haar(2^\lambda)} \left[ \Pr[(\caV^{\caG,\caG^\dagger}  \leftrightarrows (\caP^*)^{\caG,\caG^\dagger} ) (x) = 1)] \ge s + \delta \right] \le \delta.
			\]
		}
	\end{description}
\end{conjecture}

Comment: The difference between this conjecture and our main theorem (\autoref{theo:Pisuc-c-and-s}) is that in the conjecture we require soundness against all unbounded provers instead of only against semi-honest provers.
%\lnote{It should be $\caP^*$ instead of $\caP$. I noticed that~\autoref{theo:Pisuc-c-and-s}'s soundness against semi-honest part has some issue and fixed it.}

\subsection{Discussions}\label{sec:suc-discussions}

We remark that (1) the constant soundness in \autoref{conj:soundness} and the constant completeness in~\autoref{theo:Pisuc-c-and-s} can be easily amplified to $n^{-\omega(1)}$ and $1-n^{-\omega(1)}$ by repeating the protocols $\log^2 n$ times, and (2) assuming $\QPCP$, the protocol works for all languages in $\QMA$.

\begin{corollary}\label{cor:suc-protocol-for-LocalQMA}
	Assuming \autoref{conj:soundness}, there is a protocol for $L \in (k,\gamma)\text{-}\LocalQMA_{s,c}$ with $\lambda \cdot \polylog(n)$ communication complexity, completeness $1-n^{-\omega(1)}$ and soundness $n^{-\omega(1)}$ in $\QHROM$. Also, if $V_L$ is strongly explicit, then the verifier running time of the protocol is also bounded by $\lambda \cdot \polylog(n)$.
	
	Moreover, if we further assume that $\QPCP$ holds, then the aforementioned succinct protocol exists for every $L \in \QMA$.
\end{corollary}

How easy is it for the prover to cheat after having sent the commitment to the verifier in the quantum Merkle tree construction? We believe (but cannot yet prove; see Conjecture~\ref{conj:soundness}) that a computationally bounded prover will not be able to make the verifier accept. However, a computationally unbounded prover can. We now demonstrate this by the application of Schrödinger–HJW theorem ~\cite{hughston1993complete} to the simple toy example of a Merkle tree with depth one. Suppose $\spz{\psi}$ is the $2b$-qubit state $\caP$ initially committed to, and $\spz{\phi}$ is another $2b$-qubit state that $\caP$ wishes to cheat by switching $\spz{\psi}$ with. Mathematically the process of committing, switching the initial state and lastly decommitting writes 
\[
\caG^{\dagger} (W \otimes I) \caG (\spz{\psi} \otimes \spz{0^b})  \approx (\spz{\phi} \otimes \spz{0^b}),
\]
where in the above we think of $\spz{0^b}$ as the parent and see that the initially committed $2b$-qubit state $\spz{\psi}$ can be changed to another completely different $2b$-qubit state $\spz{\phi}$ by applying Schrödinger-HJW theorem (\ie, application of a $W \otimes I$) for a suitable unitary $W$ that acts on the first $2b$ qubits. We note that such $W$ exists, because the reduced density matrix of the last $b$ qubits of both $\caG (\spz{\psi} \otimes \spz{0^b})$ and  $\caG (\spz{\phi} \otimes \spz{0^b})$ are very close to the maximally mixed state, for any two fixed states $\spz{\phi}$ and $\spz{\psi}$. However, we conjecture that finding $W$ requires computationally unbounded prover. For example, in the foregoing equation a direct way to solve for $W$ would require solving a linear system of equations that is exponentially large. Moreover, the oracle $\caG$ is fully random and does not afford any structure we can utilize to reduce the computation.

This is exacerbated by the fact that finding a unitary $W$ that makes the two sides approximately equal can also make the verifier accept with sufficiently high probability. We leave this resolution as a mathematical challenge.
	
	\section*{Acknowledgments} L.C. would like to thank Jiahui Liu and Qipeng Liu for helpful discussions and pointing out many related works. This work was done while L.C. did an internship at IBM Quantum Research. L. C. is supported by NSF CCF-2127597 and an IBM Fellowship.
	
	\bibliographystyle{quantum}
	\bibliography{literature}
	
\end{document}